\documentclass{article}
\usepackage{amsmath,amssymb,amsfonts,amsthm,mathtools,xcolor}
\usepackage[hidelinks,pdfusetitle]{hyperref}
\usepackage{tikz}

\newtheorem{theorem}{Theorem}
\numberwithin{theorem}{section}
\newtheorem{proposition}[theorem]{Proposition}
\newtheorem{lemma}[theorem]{Lemma}

\theoremstyle{remark}
\newtheorem{remark}[theorem]{Remark}
\numberwithin{equation}{section}

\newcommand{\unquad}{\hspace{-1em}}

\newcommand{\Nsusy}{\mathcal{N}}
\newcommand{\CC}{\mathbb{C}}
\newcommand{\RR}{\mathbb{R}}
\newcommand{\QQ}{\mathbb{Q}}
\newcommand{\ZZ}{\mathbb{Z}}

\newcommand{\eps}{\epsilon}

\newcommand{\bei}{\begin{itemize}}\newcommand{\eei}{\end{itemize}}
\newcommand{\bel}{\be\label}
\newcommand{\be}{\begin{equation}}\newcommand{\ee}{\end{equation}}
\newcommand\beal{}
\def\beal#1\ee{\begin{equation}\begin{aligned}\relax #1\end{aligned}\end{equation}}
\newcommand\beall{}
\def\beall#1#2\ee{\begin{equation}\label{#1}\begin{aligned}\relax #2\end{aligned}\end{equation}}
\newcommand\bega{}
\def\bega#1\ee{\begin{equation}\begin{gathered}\relax#1\end{gathered}\end{equation}}
\newcommand{\Zinst}{Z_{\textnormal{inst}}}

\newcommand{\zbif}{z_{\textnormal{bif}}}
\newcommand{\zadj}{z_{\textnormal{adj}}}
\newcommand{\zvec}{z_{\textnormal{vec}}}
\newcommand{\zfun}{z_{\textnormal{fun}}}
\newcommand{\mfun}{m^{\textnormal{fun}}}
\newcommand{\mbif}{m^{\textnormal{bif}}}
\newcommand{\at}{\tilde{a}}
\newcommand{\bt}{\tilde{b}}
\newcommand{\Yt}{\widetilde{Y}}
\newcommand{\Nt}{\widetilde{N}}
\newcommand{\iti}{\widetilde{\imath}}
\newcommand{\jt}{\widetilde{\jmath}}

\DeclareMathOperator{\dist}{dist}
\newcommand{\ceil}[1]{\lceil#1\rceil}
\newcommand{\bigceil}[1]{\bigl\lceil#1\bigr\rceil}

\newcommand{\floor}[1]{\lfloor#1\rfloor}
\newcommand{\bigfloor}[1]{\bigl\lfloor#1\bigr\rfloor}

\newcommand{\mathtikz}[2][]
  {\ensuremath{\vcenter{\hbox{%
          \begin{tikzpicture}[#1]#2\end{tikzpicture}}}}}
\newcommand{\quiver}[2][]
  {\mathtikz[semithick,node distance=3em,#1]{#2}}
\tikzset{color-group/.style = {
    shape = circle,
    minimum size = 2.5ex,
    inner sep = .5ex,
    draw}}
\tikzset{flavor-group/.style = {
    shape = rectangle,
    minimum size = 2.5ex,
    inner sep = .5ex,
    draw}}

\begin{document}

\title{Convergence of Nekrasov instanton sum for unitary quivers}
\author{Bruno Le Floch}
\date{June 2026}

\maketitle

\begin{abstract}
The convergence radius of Nekrasov partition functions (as a function of instanton counting parameters) is shown to be positive for 4d $\mathcal{N}=2$ quiver gauge theories with unitary gauge groups in an open dense subset of parameters.  For $U(N)$ SQCD this is established if the ratio of equivariant parameters $b^2=\epsilon_1/\epsilon_2$ belongs to $\CC\setminus[0,+\infty)$ and Coulomb parameters or masses are away from a lattice of hyperplanes.  For general quivers it is only established for $b^2\in\CC\setminus\RR$.  When gauge multiplets are asymptotically free, the radius is infinite, whereas in the (mass-deformed) conformal case the radius admits a positive lower bound that only depends on $b^2$. The proof relies on the expression of the partition function as a sum over tuples of partitions, and a proof of absolute convergence based on combinatorial inequalities on products of (co)hook lengths. Through the AGT correspondence this implies that large classes of Virasoro and W-algebra conformal blocks on the sphere or torus have positive convergence radius, for generic dimensions and complex central charges.
\end{abstract}

\tableofcontents

\section{Introduction}

\subsection{Quivers and Nekrasov partition functions}

The realm of 4d $\Nsusy=2$ supersymmetric gauge theories has greatly improved out collective understanding of nonperturbative quantum field theory throughout the decades.  A major strand of research has been to derive exact results regarding the low-energy behaviour of such theories, starting from the seminal work of Seiberg and Witten~\cite{hep-th/9407087} which obtained the prepotential of SQCD from holomorphicity and asymptotic conditions.  This was later derived microscopically by Nekrasov~\cite{hep-th/0206161} who defined and evaluated the (now named) Nekrasov partition function of some 4d $\Nsusy=2$ gauge theories on the Omega background.  This consists in a \emph{formal power series} over an instanton counting parameter~$q$, whose coefficients are equivariant integrals over instanton moduli spaces.  These depend on the parameters $\eps_1,\eps_2$ of the Omega background, as well as mass and Coulomb branch parameters.

The AGT correspondence~\cite{0906.3219} relates the Nekrasov partition function to Virasoro and W-algebras conformal blocks.  In this translation, $q$~measures cross-ratios of operator positions, and the power series amounts to the OPE expansion, which is expected to converge, with a specific convergence radius.
The convergence of Nekrasov partition functions, or of conformal blocks, has been shown to converge in three settings (see also \cite{1608.02566,1709.05232} in 5d):
\begin{itemize}

\item for conformal blocks appearing in rational CFTs such as minimal models ($C_2$-cofinite VOA and $C_1$-cofinite modules), which have a parameter $b^2=\eps_1/\eps_2\in(-\infty,0)\cap\QQ$ and discrete masses~$m$ and Coulomb branch parameters~$a$, cf.~\cite{2204.04409} and references therein.

\item for the unrefined case $b^2=-1$, where the conformal blocks admit a free-field realization ($c=1$ in the Virasoro case), and convergence has been proven in many cases (4d $\Nsusy=2$ linear and circular quiver theories) in the study of (Fourier-transforms of) isomonodromic tau functions \cite{1403.1235,1608.00958,1705.01869,1712.08546,2011.06292} and directly~\cite{2212.06741};

\item 4d $\Nsusy=2^*$ $U(N)$ theory (vector multiplet and adjoint hypermultiplet) in the case $b^2\in(0,1)\cup(1,+\infty)$, where the probabilistic construction of Liouville CFT was used to prove a positive convergence radius for some range of parameters~$m,a$ \cite{2003.03802}, and in the case $b^2\in\CC\setminus[0,+\infty)$ for generic $m,a$, where the convergence radius was shown to be positive in \cite{2212.06741} and equal to~$1$ in~\cite{2602.19425} (with additional results for $b^2>0$);

\end{itemize}
The present work provides the first \emph{generic} convergence results beyond the ${\Nsusy=2^*}$ theory.  Specifically, we\footnote{By a mild abuse of notation, the author will use plural pronouns instead of~I\@.  This should not lead to a confusion in any of the main results since they do not feature anything singular.  Relatedly, $\sqrt{-1}$ will be used instead of~$i$ for the imaginary unit, to avoid confusion with Young diagram rows.} consider 4d $\Nsusy=2$ $U(N)$ SQCD with $N_f\leq 2N$ fundamental hypermultiplets, and more generally the class of 4d $\Nsusy=2$ quiver gauge theories described by a collection of vertices $V_0=\{1,\dots,K\}$ and edges~$V_1$ with
\begin{itemize}
\item a gauge group\footnote{As is usual in the study of Nekrasov partition functions, we consider the theories with $U(N_I)$ gauge group rather than $SU(N_I)$.
We will not discuss how the $U(1)$ contribution decouples to give $SU(N_I)$ instanton partition functions, as it is outside the scope of this work.} $U(N_I)$ for each vertex $I\in V_0$,
\item a bifundamental hypermultiplet of $U(N_I)\times U(N_J)$ of mass $\mbif_{IJ}$ for each edge $(I,J)\in V_1$,
\item $M_I\geq 0$ fundamental hypermultiplets of $U(N_I)$ of masses $\mfun_{I\alpha}$ for $1\leq I\leq K$ and $1\leq\alpha \leq M_I$,
\end{itemize}
under the \emph{balance condition}
\bel{balance-ineq}
b_I \coloneqq 2N_I - M_I - \sum_{(I,J)\in V_1} N_J - \sum_{(J,I)\in V_1} N_J \geq 0 , \qquad I\in V_0 .
\ee
To avoid clutter we do not include extra labels that would be needed to account for quivers with more than one edge between two given nodes.
When $V_1$ includes a loop $(I,I)$ the corresponding hypermultiplet is actually an adjoint of~$U(N_I)$, and all the formulas go through.
This class of quivers are well-known to admit an ADE and affine ADE classification, which we will not use the details of.

A given gauge group factor $U(N_I)$ is asymptotically free if the inequality~\eqref{balance-ineq} is strict, whereas the coupling does not run if the node is balanced, namely if the inequality~\eqref{balance-ineq} is saturated.
This class of theories can be depicted by quivers (cf.~\autoref{fig:quiver}) and admit brane realizations in IIA string theory whose uplift to M theory places these theories as important members of the class~S uncovered by Gaiotto~\cite{0904.2715}.

The Nekrasov partition function takes the form
\bel{Zinst-orig}
\Zinst(m,a; q) = \sum_{n_1,\dots,n_K\geq 0} q_1^{n_1} \dotsm q_K^{n_K} Z_n(m,a) ,
\ee
with an instanton-counting parameter $q_I$ for $1\leq I\leq K$, where each $Z_n(m,a)$ is an equivariant integral over (the product over $1\leq I\leq K$ of) the instanton moduli space of $n_I$ instantons of~$U(N_I)$.
Equivariant localization reduces each such integral to a finite-dimensional matrix model, which reduces further to a finite sum of residues, labeled by collections of Young diagrams,
\begin{subequations}
\begin{align}
\Zinst(m,a; q)
  & = \sum_{n_1,\dots,n_K\geq 0} q_1^{n_1} \dotsm q_K^{n_K} \sum_{|Y_1|=n_1,\dots,|Y_K|=n_K} Z_Y(m,a)
    \label{Zinst-Young-k}
  \\
  & = \sum_{Y_1,\dots,Y_K} q_1^{|Y_1|} \dotsm q_K^{|Y_K|} Z_Y(m,a) .
    \label{Zinst-Young-all}
\end{align}
\end{subequations}
Here, each $Y_I=(Y_{I1},\dots,Y_{IN_I})$ consists of $N_I$ Young diagrams and $|Y_I|=\sum_{1\leq\alpha\leq N_I}|Y_{I\alpha}|$ is the total number of boxes in~$Y_I$.
The coefficients $Z_Y(m,a)$ are rational functions\footnote{The explicit expressions are given in \autoref{subsec:Zinst} below.  We take these formulas as definitions of the instanton partition functions, and do not discuss their origin as equivariant integrals over instanton moduli spaces, evaluated rigorously in~\cite{math/0609841}.} of masses~$m$ and Coulomb parameters~$a$ that involve products over boxes of the Young diagrams~$Y_{I\alpha}$.

\begin{figure}
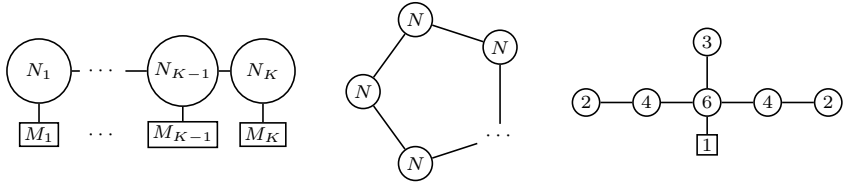
\centering
\begin{scriptsize}
  \quiver[color-group/.append style={minimum size=8ex, node distance=10ex}]{%
    \node (N1) [color-group] {$N_1$};
    \node (Ndots)[right of=N1] {$\cdots$};
    \node (Nn-1)[color-group, right of=Ndots] {$N_{K-1}$};
    \node (Nn) [color-group, right of=Nn-1] {$N_K$};
    \node (M1) [flavor-group, below of=N1] {$M_1$};
    \node (Mdots)[below of=Ndots] {$\cdots$};
    \node (Mn-1) [flavor-group, below of=Nn-1] {$M_{K-1}$};
    \node (Mn) [flavor-group, below of=Nn] {$M_K$};
    \draw (Nn)   to (Nn-1);
    \draw (Nn-1) to (Ndots);
    \draw (Ndots) to (N1);
    \draw (Nn) -- (Mn);
    \draw (Nn-1) -- (Mn-1);
    \draw (N1) -- (M1);
  }
  \qquad
  \quiver{%
    \node (N1) at (-180:1) [color-group] {$N$};
    \node (N2) at (-108:1) [color-group] {$N$};
    \node (Ndots) at (-36:1) {$\cdots$};
    \node (Nn-1) at (36:1) [color-group] {$N$};
    \node (Nn) at (108:1) [color-group] {$N$};
    \draw (Nn)   to (Nn-1);
    \draw (Nn-1) to (Ndots);
    \draw (Ndots) to (N2);
    \draw (N2) to (N1);
    \draw (N1) to (Nn);
  }
  \qquad
  \quiver[scale=.8]{%
    \node (N1) at (0,0) [color-group] {$2$};
    \node (N2) at (1,0) [color-group] {$4$};
    \node (N3) at (2,0) [color-group] {$6$};
    \node (N4) at (3,0) [color-group] {$4$};
    \node (N5) at (4,0) [color-group] {$2$};
    \node (N6) at (2,1) [color-group] {$3$};
    \node (M) at (2,-.7) [flavor-group] {$1$};
    \draw (N1)--(N2);
    \draw (N2)--(N3);
    \draw (N3)--(N4);
    \draw (N4)--(N5);
    \draw (N3)--(N6);
    \draw (N3)--(M);
  }
\end{scriptsize}
\caption{\label{fig:quiver}Linear, circular, and (an example) E-type quiver gauge theories with unitary gauge groups.  In the circular case the balance condition forces ranks to be equal and prevents the addition of fundamental hypermultiplets.}
\end{figure}



Following the approach in~\cite{2212.06741}, our technique consists of bounding each term individually, then studying the sum of these (positive) bounds.
In other words, we replace the question of convergence of the sum over~$n$ in~\eqref{Zinst-Young-k} by the more constraining question of absolute convergence as a sum over all collections~$Y$, namely:
\bel{absolute-sum}
\text{When is } \sum_Y |q_1|^{|Y_1|}\dots|q_K|^{|Y_K|} |Z_Y(m,a)|<+\infty \text{ ?}
\ee
We dub the interior\footnote{This is similar to how one considers \emph{open} convergence disks for one-variable analytic functions even when their Taylor series converges absolutely on the disk boundary, such as what happens for polylogarithms.} of this region of~$q$ the \emph{absolute convergence domain,} and call this sum the \emph{absolute Nekrasov instanton sum.}
The absolute convergence for some~$q$ implies the convergence of~$\Zinst$, but the converse does not hold since phases can lead to cancellations for each fixed instanton number.  In fact, CFT considerations suggest that the convergence of~$\Zinst$ should only depend on the cross-ratios~$q$ and not on other parameters, whereas the absolute convergence results that we find depend on~$b^2$ and degrade as $b^2$ approaches the real line.

An important advantage of absolute convergence (when it holds) is that it allows any reorganization of terms in~\eqref{Zinst-Young-all}, which justifies some resummations that have been freely performed in other works on Nekrasov partition functions.

\subsection{Statement of the convergence result}

Throughout our discussion we assume non-zero deformation parameters and we phrase our conclusions in terms of the Liouville/Toda coupling constant~$b$ defined (up to a sign) by
\be
\eps_1,\eps_2\neq 0 , \qquad
b^2 = \eps_1 / \eps_2 .
\ee
To state the genericity condition on masses and Coulomb branch parameters, it is useful to introduce the lattices
\beal
\Lambda(\eps_1,\eps_2) & = \ZZ \eps_1 + \ZZ \eps_2 = \bigl\{k\eps_1 + l\eps_2 \bigm| k,l\in\ZZ \bigr\} \subset \CC , \\
\Lambda_{\geq}(\eps_1,\eps_2) & = \ZZ_{\geq 0} \eps_1 + \ZZ_{\geq 0} \eps_2 = \bigl\{k\eps_1 + l\eps_2 \bigm| k,l\in\ZZ, k,l\geq 0 \bigr\} \subset \CC .
\ee
For real $b^2$, the lattice $\Lambda(\eps_1,\eps_2)$ degenerates to a discrete (if $b^2\in\QQ$) or dense subset of a line, and in particular has (infinitely many) combinations $m\eps_1+n\eps_2$ arbitrarily close to zero, whereas $\Lambda_{\geq}(\eps_1,\eps_2)$ only does so for $b^2<0$.
It is also useful to introduce the closures $\overline{\Lambda(\eps_1,\eps_2)}$ and $\overline{\Lambda_{\geq}(\eps_1,\eps_2)}$, which coincide with the lattices except for real irrational~$b^2$ and negative real irrational~$b^2$, respectively.

We also need the notation $\floor{x}$ (resp.\ $\ceil{x}$) for the largest (resp.\ smallest) integer $\leq x$ (resp.\ $\geq x$), and
\be
\{x\} = x - \floor{x} \in [0,1)
\ee
for the fractional part.
Coulomb branch parameters are denoted by $a_{I\alpha}$ for $1\leq I\leq K$ and $1\leq\alpha\leq K$, and we omit the ranges of $I,\alpha,\beta$ when they are clear from context.

\begin{theorem}[On quivers]\label{thm:quivers}
  Assume that $b^2\in\CC\setminus\RR$ and that none of the terms $a_{I\alpha}-a_{I\beta}$ for $\alpha\neq\beta$, nor the terms $a_{I\alpha}-a_{J\beta}-\mbif_{IJ}$ for $(I,J)\in V_1$, nor the terms $a_{I\alpha}-\mfun_{I\beta}$ belong to the lattice $\Lambda(\eps_1,\eps_2)$.
  Then there exists $R>0$ (depending only on~$b^2$) such that the absolute sum converges as long as $|q_I|<R$ for each node~$I$ that is balanced (cf.~\eqref{balance-ineq}).  Instanton counting parameters~$q_I$ of unbalanced nodes can be arbitrary (namely the convergence radius is infinite).
\end{theorem}

Observe that the result specializes easily to SQCD with $N_f\leq 2N$, for which there is an infinite convergence radius in the asymptotically free case ($<2N$ flavours) for all $b^2\in\CC\setminus[0,+\infty)$, and in the mass-deformed conformal case ($2N$~flavours) a convergence radius $R>0$ that depends on $b^2\in\CC\setminus\RR$.
In this theory we actually achieve a positive radius also in the case $b^2<0$, as summarized next (with indices $I$~dropped).

\begin{theorem}[On SQCD]\label{thm:SQCD}
  Consider $N_f\leq 2N$ SQCD\@. Assume that $b^2\in\CC\setminus [0,+\infty)$ and that none of the terms $a_\alpha-a_\beta$ for $\alpha\neq\beta$ nor the terms $a_\alpha-m_\beta$ belong to the closure of the lattice $\overline{\Lambda(\eps_1,\eps_2)}$.
  Then the absolute Nekrasov instanton sum has a positive convergence radius $R>0$ that depends only on~$b^2$ and is infinite for $N_f<2N$.
\end{theorem}

\begin{remark}
  We emphasize once more that the divergence of~\eqref{absolute-sum} for some ranges of parameters does not give any information on whether or not the first expression in~\eqref{Zinst-Young-k} is well-defined and convergent for all $|q|<1$.  In fact, if expectations about Virasoro/$W_N$ conformal bootstrap are correct, then the convergence domains should be universal, with no $b^2$~dependence.  This requires the sum over~$Y$ with fixed $|Y_1|,\dots,|Y_K|$ to feature cancellations.
\end{remark}

\begin{remark}
  The convergence radius is worse here than for the 4d $\Nsusy=2^*$ theory treated in~\cite{2602.19425}, where $R=1$ for all $b^2\in\CC\setminus[0,+\infty)$.
  In the present paper on SQCD and quivers the optimal radius is not obtained, partly because some technical inequalities are not optimal, and partly because \emph{absolute} convergence likely fails, driven by partitions with an interesting limit shape \`a la Nekrasov--Okounkov.
\end{remark}

\subsection{Proof strategy and outline}

We now outline the main ideas of the proof, with references to the intermediate results in the main file.  In contrast to~\cite{2602.19425}, we make no attempt to place rigorous upper bounds on the convergence radius, focusing instead on proving convergence in domains where this can be done.  In \autoref{sec:no-mass} we find the following simplifications.
\begin{itemize}
\item All factors are uniformly bounded away from zero and above by the total number~$n$ of instantons, so we can freely throw away $o(n/\log n)$ of them without affecting the (open) absolute convergence domain (Proposition~\ref{prop:removing-factors}).  This allows being somewhat coarse when separating zones in Young diagrams, which helps extend previous work~\cite{2212.06741} beyond the unrefined limit.

\item To show that the convergence radius only depends on $b^2$ and not on mass and Coulomb branch parameters, the key is to correctly take into account, in all factors, lower bounds on the spread of instantons.
  We find that any given bounded subset of the complex plane can accommodate at most $O(\sqrt{n})$ instantons, so that typical factors in $Z_Y(m,a)$ are of order at least $\sqrt{n}$.  Constant shifts by mass and Coulomb branch parameters thus have a negligible effect on most factors and do not affect the absolute convergence domain (cf.\ Proposition~\ref{prop:irrelevant-masses}).
\end{itemize}

We then consider SQCD with $b^2\in(-\infty,0)$ in \autoref{sec:SQCD-neg}, by extending the proof of the unrefined case $b^2=-1$ in~\cite{2212.06741}.  There are two main steps.
\begin{itemize}
\item In the vector multiplet contribution, off-diagonal factors are bounded by diagonal ones times a power $g(b^2)^n$ of some function~$g$, by following exactly the proof of Lemma 3.9 in~\cite{2212.06741} which pairs up boxes that have close enough contributions in the numerator and denominator.  This converts the problem to the $U(1)$ case.

\item In the unrefined case, the remaining ratio of contents by hook product was bounded through representation-theoretic exact expressions.  Here, we take a more pedestrian tack by pairing numerator and denominator factors, with each ratio being less than one.
\end{itemize}

We move on in \autoref{sec:quivers} to non-real $b^2\in\CC\setminus\RR$ for a general quiver (including SQCD\@).  Here there are two types of ratios to bound.
\begin{itemize}
\item A ratio of diagonal and off-diagonal vector multiplet contributions, which involves pairs of Young diagrams $\lambda,\mu$, and also deals with the bifundamental contribution.  By separating products into the intersection $\lambda\cap\mu$ and the symmetric difference $\lambda\Delta\mu$, we reduce the problem to a comparison of products of hook and cohook lengths proven by Morales, Pak, and Panova~\cite{1610.07561}.

\item The problem reduces to the $U(1)$ case, and the ratio of fundamental matter and vector contributions are easily bounded by the same hook/cohook bound.
\end{itemize}

\section{Mass and Coulomb parameters do not matter}
\label{sec:no-mass}

\subsection{Instanton partition function}
\label{subsec:Zinst}

We follow the same conventions on Young diagrams and instanton partition functions as \cite{hep-th/0211108,2212.06741}, but do not use the arm/leg length notation since it is useful in our setting to see more explicitly the different contributions.  Boxes of a Young diagram~$\lambda$ are labeled by $(i,j)$ with $1\leq i\leq\lambda_j$ or equivalently $1\leq j\leq\lambda'_i$, with $\lambda'$ the transposed diagram.
We define
\be
E(a,\lambda,\mu;i,j) = a + \eps_2 - \eps_1 (\mu'_i - j) + \eps_2 (\lambda_j - i)
\ee
from which we express the contribution of a bifundamental of $U(N)\times U(\Nt)$ of mass~$m$, of an adjoint, and of a vector multiplet as
\begin{align}
\zbif(a,Y,\at,\Yt;m)
& =
\begin{aligned}[t]
& \prod_{\alpha=1}^{N} \prod_{\beta=1}^{\Nt} \biggl( \prod_{(i,j)\in Y_\alpha} E(a_\alpha - \at_\beta - m, Y_\alpha, \Yt_\beta; i,j) \\
& \prod_{(i,j)\in \Yt_\beta} \bigl(- E(\at_\beta - a_\alpha + m-\eps_1 - \eps_2, \Yt_\beta, Y_\alpha; i,j)\bigr) \biggr) ,
\end{aligned}
\\
\zadj(a,Y;m) & = \zbif(a,Y,a,Y;m) ,
\\
\zvec(a,Y) & = (\zadj(a,Y;0))^{-1} .
\end{align}
A collection of $M\geq 0$ fundamental hypermultiplets of $U(N)$ contributes the same as a bifundamental with one of the collection of Young diagrams vanishing, namely $\Yt=(\emptyset,\dots,\emptyset)$.  It is usually written differently by exchanging boxes $(i,j)$ and $(1+\lambda_j-i,j)$:
\beall{zfun}
\zfun(a,Y;\mfun)
& = \zbif(a,Y,\mfun,\emptyset;0)
\\
& = \prod_{\alpha=1}^{N} \prod_{\beta=1}^{M} \prod_{(i,j)\in Y_\alpha} E(a_\alpha - \mfun_\beta , Y_\alpha, \emptyset; i,j)
\\
& = \prod_{\alpha=1}^{N} \prod_{\beta=1}^{M} \prod_{(i,j)\in Y_\alpha} \bigl(
a_\alpha - \mfun_\beta + \eps_1 j + \eps_2 i \bigr) .
\ee

The instanton partition function of the quivers of interest to us are
\beall{Zinst-detail}
\Zinst & = \sum_{Y_1,\dots,Y_K} q_1^{|Y_1|} \dots q_K^{|Y_K|} Z_Y(m,a) ,
\\
Z_Y(m,a) & = \prod_{I=1}^K \bigl(\zvec(a_I,Y_I) \zfun(a,Y;\mfun_I)\bigr) \prod_{(I,J)\in V_1} \zbif(a_I,Y_I, a_J,Y_J;\mbif_{IJ})
\ee

\subsection{Small factors can be ignored}

We find here that absolute convergence is not affected much\footnote{Only the convergence on the boundary of the absolute convergence domain is affected, but since we do not achieve the optimal convergence radius, this is moot.} by removing a small proportion of the factors (cf.\ Lemma~\ref{lem:bound-E} and Proposition~\ref{prop:removing-factors}), and there are sufficiently few factors smaller than any fixed bound (cf.\ Lemma~\ref{lem:bound-density}).  Later sections will thus be able to focus on large scale features of Young diagrams.

\paragraph{Ignoring sufficiently few factors}

\begin{lemma}\label{lem:bound-E}
  Under our genericity assumptions, all factors $E(a,\lambda,\mu;i,j)$ appearing in the contribution $Z_Y(m,a)$ of the instanton partition function~\eqref{Zinst-detail} belong to an annulus
  \be
  1/C \leq |E(a,\lambda,\mu;i,j)| \leq C |Y|
  \ee
  for some $C>1$ that only depends on masses, Coulomb branch parameters, and equivariant parameters~$\eps_1,\eps_2$, where $|Y|$~is the total instanton number.
\end{lemma}

\begin{proof}
First, Young diagram rows/columns are bounded by the total number of boxes, so $|E(a,\lambda,\mu;i,j)| \leq C (|\lambda|+|\mu|)$ for all $(i,j)\in\lambda$, where $C>1$ does not depend on $\lambda,\mu,i,j$ (the claim is vacuuous for $\lambda=\emptyset$).  This proves the upper bound.  The lower bound relies on our genericity assumptions.  All factors except the diagonal vector multiplets take the form $a+k\eps_1+l\eps_2$ for some $k,l\in\ZZ$ and $a$ consisting of differences of Coulomb branch or mass parameters, which are assumed to stay away from the closure $\overline{\Lambda(\eps_1,\eps_2)}$.  The diagonal vector multiplet contributions are, for $(i,j)\in\lambda$ being one of the Young diagrams~$Y_{I\alpha}$,
\beal
E(0,\lambda,\lambda;i,j) & = - \eps_1(\lambda'_i - j) + \eps_2 (\lambda_j - i + 1) ,
\\
E(-\eps_1-\eps_2,\lambda,\lambda;i,j) & = - \eps_1(\lambda'_i - j + 1) + \eps_2 (\lambda_j - i) .
\ee
These take the form $-k\eps_1+l\eps_2$ for $k,l\geq 0$ and not both zero.
Under our assumption $b^2 = \eps_1/\eps_2 \in \CC\setminus [0,+\infty)$ the set of such $-k\eps_1+l\eps_2$ is bounded away from zero.
\end{proof}

\begin{proposition}\label{prop:removing-factors}
  Consider a product $Z_Y^\circ(m,a)$ defined by removing some factors $E(\dots)$ from the factors $\zvec,\zfun,\zbif$ of~$Z_Y(m,a)$ listed in~\eqref{Zinst-detail}, and consider the series
  \be
  \Zinst^\circ = \sum_{Y_1,\dots,Y_K} q_1^{|Y_1|} \dots q_K^{|Y_K|} Z_Y^\circ(m,a) .
  \ee
  If the number of removed factors is bounded as $o(n/\log n)$ where $n=|Y|$ is the total instanton number, then the series $\Zinst^\circ$ has the same absolute convergence domain as the Nekrasov instanton sum~$\Zinst$.
\end{proposition}

\begin{proof}
  All factors $|E(\dots)|$ involved in $|Z_Y(m,a)|$ belong to an interval $[1/C,Cn]$ as stated in Lemma~\ref{lem:bound-E}.  Denoting by $p$ the number of factors being removed, the product of factors that are being removed thus belongs to an interval $[C^{-p},(Cn)^p]$.  If $p=o(n/\log n)$ then this interval can also be described as $\exp(o(n))$, which can be compensated by an arbitrarily small reduction in the modulus of instanton counting parameters.  Specifically, if the Nekrasov instanton sum expression of~$\Zinst$ converges absolutely for some $q_1,\dots,q_K$, then the analogous expression of $\Zinst^\circ$ converges absolutely for any $\tilde q_1,\dots,\tilde q_K$ with each $|\tilde q_I|<|q_I|$, and the same holds with $\Zinst$ and $\Zinst^\circ$ swapped.  The regions of absolute convergence of $\Zinst$ and of $\Zinst^\circ$ each contain the interior of the other one, so their interior (the absolute convergence domain) is equal.
\end{proof}

\paragraph{Factors are not too dense}

We now state how many factors $E(\cdots)$ can be within a given bounded set.  For definiteness we state it for disks, but it holds for any set, with $r$ replaced by the set's diameter.

\begin{lemma}\label{lem:bound-density}
  For $\eps_1,\eps_2\in\CC\setminus\{0\}$ with $b^2=\eps_1/\eps_2\in\CC\setminus[0,+\infty)$, there exists a constant $C>0$, such that for any $A\in\CC$ and any disk of center $a\in\CC$ and radius $r>0$, and any pair of partitions $\lambda,\mu$, the number of boxes $(i,j)\in\lambda$ with $E(0,\lambda,\mu;i,j)$ in that disk admits the following bound,
  \be
  \#\bigl\{(i,j)\in\lambda \bigm| |E(A,\lambda,\mu;i,j) - a| < r \bigr\} \leq C (r+1) \sqrt{|\lambda|} .
  \ee
\end{lemma}

\begin{proof}
  The parameter~$A$ can be immediately set to $A=0$ by shifting~$a$.

  Since $b^2$ is not positive, its square-roots are not real.  We fix for definiteness a square-root $b$ with $\Im b>0$, and correspondingly $\sqrt{\eps_1\eps_2} = \eps_2 b$.
  We introduce
  \be
  \gamma = \min(\Im b,-\Im(1/b)) > 0 .
  \ee
  We divide throughout by~$\sqrt{\eps_1\eps_2}$.  The disk of radius $r/|\eps_1\eps_2|^{1/2}$ and centered at $a/\sqrt{\eps_1\eps_2}$ is covered by $O(r+1)$ horizontal strips of width~$\gamma$ defined by
  \be
  \Sigma_p \coloneqq \{z\mid p\leq\Im z/\gamma<p+1\} \subset \CC .
  \ee
  The condition for the contribution of $(i,j)\in\lambda$ to belong to this strip, namely $E(0,\lambda,\mu;i,j)/\sqrt{\eps_1\eps_2}\in\Sigma_p$, is conveniently rewritten by separating the $i$ and $j$ dependence,
  \be
  p \leq
  \frac{\Im b}{\gamma} \Bigl( j - \frac{\Im(-1/b)}{\Im b} \lambda_j \Bigr)
  + \frac{\Im(-1/b)}{\gamma} \Bigl( i - 1 - \frac{\Im b}{\Im(-1/b)} \mu'_i \Bigr)
  < p+1
  \ee
  Since $\lambda_j$ (resp.~$\mu_i$) is monotonically (weakly) decreasing in~$j$ (resp.\ in~$i$), the terms in parentheses increase by at least~$1$ when~$j$ (resp.~$i$) is incremented.  Since the coefficient $\Im b/\gamma$ and $\Im(-1/b)/\gamma$ are both $\geq 1$, the boxes $(i,j)$ satisfying these bounds (for a given~$p$) must all lie in different rows and different columns of~$\lambda$.  Denoting by $q_p$ the number of such boxes, the partition~$\lambda$ must thus contain the partition $(q_p,q_p-1,\dots,1)$, hence $q_p(q_p+1)/2 \leq |\lambda|$ and finally $q_p\leq \sqrt{2|\lambda|}$.  Summing over the strips $\Sigma_p$ that cover the appropriate disk gives the desired conclusion.
\end{proof}

\subsection{The effect of mass parameters}

We now embark in the proof that the masses and Coulomb branch parameters (denoted below by $A,B$), or more generally any bounded shift in the various factors, do not affect convergence.
This is done by showing sub-exponential bounds on the ratio of products with or without shifts.  As in the previous section, such sub-exponential factors do not affect the (open) absolute convergence domain.



\begin{proposition}[Mass parameters are mostly irrelevant]
  \label{prop:irrelevant-masses}
  Fix $\eps_1,\eps_2\in\CC\setminus\{0\}$ with $b^2=\eps_1/\eps_2\in\CC\setminus[0,+\infty)$,
  and $A,B\in\CC$ with $B\not\in\overline{\ZZ\eps_1+\ZZ_{<0}\eps_2}$ (closure of the set of $k\eps_1+l\eps_2$ with $k,l$ integer and $l<0$).  Then there exists a constant $C>0$ such that for any pair of partitions $\lambda,\mu$ one has the subexponential bound
  \bel{eq10olambda}
  \biggl| \prod_{(i,j)\in\lambda} \frac{E(A,\lambda,\mu;i,j)}{E(B,\lambda,\mu;i,j)} \biggr|
  \leq e^{C|\lambda|^{1/2}(1+\log|\lambda|)} .
  \ee
  For the special case $\mu=\lambda$ the condition on~$B$ can be relaxed to $B\not\in\ZZ_{\geq 0}\eps_1+\ZZ_{<0}\eps_2$ (which is a closed set).  For the special case $\mu=\emptyset$ the condition on~$B$ can be relaxed to $B\not\in\overline{\ZZ_{<0}\eps_1+\ZZ_{<0}\eps_2}$ and $b^2\in\CC\setminus\{0\}$ can be arbitrary.
\end{proposition}

The proof of Proposition~\ref{prop:irrelevant-masses} occupies the rest of the subsection.
Observe that if both $A$ and $B$ are generic then both the ratios and their inverses are subexponential so that changing $A$ to $B$ will not change the absolute convergence domains of partition sums.


\paragraph{Proof for $b^2\in\CC\setminus[0,+\infty)$.}

We prove the main part of the proposition, namely the bound~\eqref{eq10olambda}.
We bound $1+X\leq\exp(X)$ to find
\bel{eq27}
\biggl| \prod_{(i,j)\in\lambda} \frac{E(A,\lambda,\mu;i,j)}{E(B,\lambda,\mu;i,j)} \biggr|
\leq \exp\biggl(\sum_{(i,j)\in\lambda} \frac{|B-A|}{E(B,\lambda,\mu;i,j)}\biggr) .
\ee
Our genericity assumption on~$B$ ensures that the distance
\bel{lower20}
d_B \coloneqq \dist\bigl(B, \ZZ\eps_1 + \ZZ_{<0}\eps_2\bigr)
= \inf_{k,l\in\ZZ,l<0} \bigl|B-k\eps_1-l\eps_2\bigr| > 0
\ee
is positive since we assumed that $B$ does not belong to the closure of the lattice.  All denominators are bounded below by this distance: indeed, for $(i,j)\in\lambda$ one has $\lambda_j-i+1>0$ so
\bel{E-geq-dB}
|E(B,\lambda,\mu;i,j)|
= \bigl|B-(\mu'_i-j)\eps_1+(\lambda_j-i+1)\eps_2\bigr| \geq d_B .
\ee
Let us order boxes $(i,j)\in\lambda$ as $(i_1,j_1),(i_2,j_2),\dots,(i_{|\lambda|},j_{|\lambda|})$ by non-decreasing values of $|E(B,\lambda,\mu;i,j)|$.
By Lemma~\ref{lem:bound-density}, the number of boxes satisfying $|E(B,\lambda,\mu;i,j)|<r$ for some $r>0$ is bounded by $C(r+1)\sqrt{|\lambda|}$, thus any $1\leq s\leq|\lambda|$ obeys
\be
s < C \sqrt{|\lambda|} \bigl(1+|E(B,\lambda,\mu;i_s,j_s)|\bigr) ,
\ee
hence by combining with~\eqref{E-geq-dB},
\be
|E(B,\lambda,\mu;i_s,j_s)| \geq \max\Bigl(d_B, \frac{s}{C \sqrt{|\lambda|}} - 1 \Bigr) .
\ee
We then split the sum in~\eqref{eq27} into contributions for $s\gtrless 2C\sqrt{|\lambda|}$ to get
\beal
\biggl| \prod_{(i,j)\in\lambda} \frac{E(A,\lambda,\mu;i,j)}{E(B,\lambda,\mu;i,j)} \biggr|
& \leq \exp\biggl(
\frac{2C|B-A|}{d_B}\sqrt{|\lambda|}
+ |B-A|C\sqrt{|\lambda|} H_{|\lambda|}
\biggr)
\\
& \leq \exp\bigl(O(|\lambda|^{1/2}(1+\log|\lambda|))\bigr)
\ee
where $H_p=1+1/2+1/3+\dots+1/p$ is the harmonic number.

\paragraph{Special cases $\mu=\lambda$ and $\mu=\emptyset$.}

When one restricts attention to a single partition $\mu=\lambda$, one gets $\mu'_i-j = \lambda'_i-j\geq 0$ for all boxes $(i,j)\in\lambda$, so that the reasoning above goes through with $d_B$ redefined to be the distance of~$B$ to $\ZZ_{\geq 0}\eps_1+\ZZ_{<0}\eps_2$.  For $b^2=\eps_1/\eps_2\in\CC\setminus[0,+\infty)$ this set is discrete so there is no need to take its closure when stating the condition on~$B$ in Proposition~\ref{prop:irrelevant-masses}.

Likewise, for $\mu=\emptyset$ one has $\mu'_i-j<0$, so $B$ has to avoid the closure of $\ZZ_{<0}\eps_1+\ZZ_{<0}\eps_2$.  This argument covers the range $b^2\in\CC\setminus[0,+\infty)$.  For the additional range $b^2\in(0,+\infty)$ claimed in Proposition~\ref{prop:irrelevant-masses}, different arguments are needed, which are given momentarily.

\paragraph{Remark: possible exponential growth for $b^2>0$.}

Before moving on to $\mu=\emptyset$, it is interesting to explain what goes wrong for $b^2>0$ with other partitions~$\mu$.
We point out a class of partitions (with $\mu=\lambda$) for which the ratio in~\eqref{eq10olambda} grows exponentially in~$|\lambda|$ for positive~$b^2$, so that changing mass or Coulomb parameters may shrink the absolute convergence domain.

We let $\floor{x}$ (resp.\ $\ceil{x}$) denote the largest (resp.\ smallest) integer $\leq x$ (resp.\ $\geq x$).
For some (large) integer $p\geq 1$, consider the partition
\be
\lambda_j = 1 + \bigfloor{b^2 (p-j)} , \qquad j = 1,\dots,p ,
\ee
whose Young diagram is a discrete approximation of a triangle with slope~$b^2$.
The conjugate partition is
\be
\lambda'_i = \max_{\lambda_j\geq i}(j)
= \max_{i - 1 \leq b^2 (p-j)}(j)
= p - \bigceil{b^{-2} (i - 1)} , \qquad i = 1,\dots,\lambda_1 .
\ee
For any box $(i,j)$ in this partition one has
\beal
\quad & \unquad \frac{1}{\eps_2} \bigl( -(\lambda'_i-j)\eps_1+(\lambda_j-i)\eps_2 \bigr)
\\
& = -\bigl(p-j - \bigceil{b^{-2} (i - 1)}\bigr) b^2
+ 1 + \bigfloor{b^2 (p-j)} - i
\\
& = \Bigl( \bigfloor{b^2 (p-j)} - b^2 (p-j) \Bigr)
+ b^2 \Bigl( \bigceil{b^{-2} (i - 1)} - b^{-2} (i-1) \Bigr)
\\
& \in (-1,b^2)
\ee
since the floor and ceiling functions differ from their argument by less than~$1$.
Then
\be
\biggl| \prod_{(i,j)\in\lambda} \frac{A-(\lambda'_i-j)\eps_1+(\lambda_j-i+1)\eps_2}{B-(\lambda'_i-j)\eps_1+(\lambda_j-i+1)\eps_2} \biggr|
\geq \biggl( \inf_{\substack{x\in(0,b^2+1)\\x\neq -B/\eps_2}} \Bigl| \frac{A+\eps_2 x}{B+\eps_2 x} \Bigr| \biggr)^{|\lambda|} .
\ee
For fixed $b^2$ and~$B$, one can choose $A$ large enough to ensure that the infimum is larger than~$1$.  With that choice, the lower bound found here grows exponentially with~$|\lambda|$ as announced, instead of having the sub-exponential bound $\exp(C|\lambda|^{1/2}(1+\log|\lambda|))$.

\paragraph{Proof for $\mu=\emptyset$ and $b^2>0$.}

We now address the final part of the Proposition, which is to extend to positive~$b^2$ the sub-exponential bound for $\mu=\emptyset$.
After relabelling $(i,j)\to(1+\lambda_j-i,j)$, \eqref{eq10olambda}~reduces to the claim
\bel{eq9olambda}
\biggl| \prod_{(i,j)\in\lambda} \frac{A+j\eps_1+i\eps_2}{B+j\eps_1+i\eps_2} \biggr|
\leq e^{C|\lambda|^{1/2} (1+\log|\lambda|)} .
\ee
Our genericity assumption on~$B$ ensures that denominators are bounded below.  Following the same strategy as above, it is enough to show that there are at most $O(\sqrt{|\lambda|})$ boxes with $B+j\eps_1+i\eps_2$ (or simply $j\eps_1+i\eps_2$) lying in any fixed bounded region.

Transposing the diagram exchanges $\eps_1,\eps_2$, hence inverts~$b^2$, so we can assume without loss of generality that $b^2\geq 1$.  After dividing by~$\eps_2$, the task is simply to show that at most $O(\sqrt{|\lambda|})$ boxes have $jb^2+i$ in a given (real) interval.  It is enough to show this for intervals $[p,p+1)$.  Two boxes $(i,j)$ with $jb^2+i\in[p,p+1)$ cannot share the same value of~$i$ nor~$j$ thanks to $b^2\geq 1$.  Thus, denoting by $q_p$ the number of boxes with $jb^2+i\in[p,p+1)$, the partition~$\lambda$ must contain the partition $(q_p,q_p-1,\dots,2,1)$, and as a result $q_p(q_p+1)/2\leq|\lambda|$.  This gives the desired bound $q_p=O(\sqrt{|\lambda|})$ on the density of factors, which implies~\eqref{eq9olambda} and concludes the proof of Proposition~\ref{prop:irrelevant-masses}.

\section{SQCD with negative \(b^2\)}
\label{sec:SQCD-neg}

In this section we fix $b^2=\eps_1/\eps_2\in(-\infty,0)$ and prove that there is a positive radius of convergence.  This includes the unrefined case $b^2=-1$.
Note that for $b^2$ rational the lattice $\Lambda(\eps_1,\eps_2)$ is discrete, and for $b^2$ irrational (and negative) the lattice is dense in $\overline{\Lambda(\eps_1,\eps_2)} = \RR\eps_1 = \RR\eps_2$.

\subsection{Vector multiplets: diagonal/off-diagonal factors}
\label{sec:SQCD-off-diag}

The vector multiplet contribution involves factors associated to pairs of indices $1\leq\alpha,\beta\leq N$.  Here we show, by following the proof of Lemma 3.9 in~\cite{2212.06741}, that each off-diagonal factor labeled by $(\alpha,\beta)$ is bounded by the corresponding diagonal factor $(\alpha,\alpha)$ up to an exponential factor (cf.\ Proposition~\ref{prop:SQCD-off}).  This boils down to the following inequality for pairs of Young diagrams.

\begin{lemma}\label{lem:SQCD-off-diag}
  For $b^2\in(-\infty,0)$ and $A,B\in\CC$ with $B\not\in\overline{\Lambda(\eps_1,\eps_2)}$, there exists a constant $C>0$ such that for any partitions $\lambda,\mu$ one has
  \bel{SQCD-off-diag}
  \prod_{(i,j)\in\lambda} \Bigl| \frac{E(A, \lambda, \lambda; i,j)}{E(B, \lambda, \mu; i,j)} \Bigr|
  \leq 8^{|\lambda|} e^{C|\lambda|^{1/2}(1+\log|\lambda|)} .
  \ee
\end{lemma}

\begin{proof}
  We have seen in Proposition~\ref{prop:irrelevant-masses} that changing the parameters $A,B$ only brings in factors of the form $\exp(C|\lambda|^{1/2}(1+\log|\lambda|))$ as long as their values avoid the relevant lattices.  More precisely it is enough to prove~\eqref{SQCD-off-diag} for some fixed $A,B$ with $A\not\in\ZZ_{\geq 0}\eps_1+\ZZ_{<0}\eps_2$ and $B\not\in\overline{\ZZ\eps_1+\ZZ\eps_2}$.  We select for simplicity
\be
A=0 , \qquad B=\beta\eps_1 ,
\ee
for some purely imaginary $\beta\in\RR\sqrt{-1}$ with $\Im\beta\in(-\infty,0)\cup(0,+\infty)$.  We keep $\beta$ non-zero to avoid singular factors.  Its value is unimportant.  Then (dividing by~$\eps_1$ for convenience)
\beall{sqcd-EE}
\frac{1}{\eps_1} E(0,\lambda,\lambda;i,j) & = - (\lambda'_i - j) + b^{-2}(\lambda_j + 1 - i) \in (-\infty,b^{-2}] ,
\\
\frac{1}{\eps_1} E(\beta\eps_1,\lambda,\mu;i,j) & = \beta - (\mu'_i - j) +  b^{-2} (\lambda_j + 1 - i) \in \beta + \RR .
\ee
In addition to $\beta$ above, we introduce the notation
\be
\alpha_i \coloneqq \mu'_i - \lambda'_i , \qquad
d_{ij} \coloneqq - (\mu'_i - j) + b^{-2}(\lambda_j + 1 - i) , \qquad (i,j)\in\lambda ,
\ee
so that the two lines in~\eqref{sqcd-EE} read $\alpha_i+d_{ij}$ and $\beta+d_{ij}$, respectively.
A particularly useful inequality (first line of~\eqref{sqcd-EE}) is
\be
\alpha_i + d_{ij} \leq b^{-2} < 0 .
\ee

We split the product~\eqref{SQCD-off-diag} into three products $P_-$, $P_0$, $P_+$ over boxes $(i,j)\in\lambda$ for which $d_{ij}$ is in $(-\infty,-1]$, $(-1,1)$, and $[1,+\infty)$, respectively:\footnote{This splitting is adapted from the case $b^2=-1$ in \cite{2212.06741}, where $d_{ij}$~is an integer so that the interval $(-1,1)$ reduces to $\{0\}$.  In that reference, $P_0$ is dubbed $B_1$ and the separation of their $B_2$ into $P_\pm$ is only performed in their Appendix~C.}
\beal
P_0 = \prod_{(i,j)\in\lambda,|d_{ij}|<1} \frac{\alpha_i+d_{ij}}{\beta+d_{ij}} ,
\qquad
P_\pm = \prod_{(i,j)\in\lambda, \, \pm d_{ij}\geq 1} \frac{\alpha_i+d_{ij}}{\beta+d_{ij}} .
\ee

Consider first the product~$P_0$.
Since every numerator in~$P_0$ is bounded by $(1-b^{-2})|\lambda|$ and every denominator is at least $|\beta|>0$ (in modulus), each ratio is $O(|\lambda|)$.  On the other hand, the number of factors in~$P_0$ can be bounded by noting that the boxes included there are those for which $|E(0,\lambda,\mu;i,j)| < |\eps_1|$.  By Lemma~\ref{lem:bound-density} there are thus $O(\sqrt{|\lambda|})$ factors.  Overall, we get the sub-exponential bound, for some constant $C>0$,
\bel{P0-bound}
|P_0| \leq e^{C\sqrt{|\lambda|} (1+\log|\lambda|)} .
\ee

We split the factors $P_\pm$ into contributions $P_\pm^{(i)}$ from the $i$-th row of~$\lambda$.  For each fixed~$i$, the quantity $d_{ij}$ is strictly increasing in~$j$, so the products $P_-^{(i)}$ and $P_+^{(i)}$ receive contributions from boxes $(i,j)$ with consecutive values $j\in [1,j_-]$ and $j\in [j_++1,\lambda'_i]$, respectively, with $0\leq j_-\leq j_+\leq\lambda'_i$.  Either one (or both) intervals may be empty (as implemented by $j_-=0$ or $j_+=\lambda'_i$, respectively).

Consider $P_-^{(i)}$.  We recall that $d_{ij}\leq -1$ and $\alpha_i+d_{ij}\leq 0$.  Each factor $(\alpha_i+d_{ij})/(\beta+d_{ij})$ is bounded as
\bel{Pm-bound-each-factor}
\Bigl| \frac{\alpha_i+d_{ij}}{\beta+d_{ij}} \Bigr|
= \frac{-\alpha_i-d_{ij}}{(|\beta|^2 + d_{ij}^2)^{1/2}}
\leq \frac{-\alpha_i-d_{ij}}{-d_{ij}} .
\ee
For rows with $\alpha_i<0$, the upper bound here is a decreasing function of $-d_{ij}\geq 1$.  On the other hand we know that $d_{i(j+1)}-d_{ij} = 1 + |b^{-2}|(\lambda_j-\lambda_{j+1}) \geq 1$, which is a lower bound on how spread out the $d_{ij}$ for $j\in[1,j_-]$ must be.  We learn that for $\alpha_i<0$,
\bel{Pmi-bound}
|P_-^{(i)}| = \prod_{j=1}^{j_-} \Bigl| \frac{\alpha_i+d_{ij}}{\beta+d_{ij}} \Bigr|
\leq \prod_{j=1}^{j_-} \frac{|\alpha_i|+|d_{ij}|}{|d_{ij}|}
\leq \prod_{d=1}^{j_-} \frac{|\alpha_i|+d}{d} = \binom{j_-+|\alpha_i|}{|\alpha_i|} .
\ee
For rows with $\alpha_i\geq 0$, the upper bound in~\eqref{Pm-bound-each-factor} is in $[0,1]$, so
\bel{Pmleq1}
|P_-^{(i)}| \leq 1 , \qquad \text{for } \alpha_i \geq 0 .
\ee

Consider $P_+^{(i)}$ next.  The product ranges over $j\in[j_++1,\lambda'_i]$ such that $d_{ij}\geq 1$.  As before, $d_{ij}$ for successive~$j$ differ by at least~$1$.
On the other hand, $\alpha_i+d_{ij}\leq 0$ so $\alpha_i<0$ and $|\alpha_i|$ bounds every $d_{ij}$ for $j$ in the interval considered here.  Thus, because $(|\alpha_i|-d)/d$ is monotonically decreasing in~$d$, we get
\be
|P_+^{(i)}|
\leq \prod_{j=j_++1}^{\lambda'_i} \frac{|\alpha_i|-d_{ij}}{d_{ij}}
\leq \prod_{d=1}^{\lambda'_i - j_+} \frac{|\alpha_i|-d}{d}
= \binom{|\alpha_i|-1}{\lambda'_i - j_+} .
\ee

The product $P_-^{(i)}P_+^{(i)}$ is then bounded as follows.
If $\alpha_i \geq 0$ then $P_+^{(i)}$ is an empty product and $|P_-^{(i)}|\leq 1$ as per~\eqref{Pmleq1}.
If $\alpha_i < 0$ (namely $\mu'_i < \lambda'_i$) then
\be
|P_-^{(i)}P_+^{(i)}|
\leq \binom{j_-+|\alpha_i|}{|\alpha_i|} \binom{|\alpha_i|-1}{\lambda'_i - j_+}
\leq 2^{j_- + 2|\alpha_i| - 1} \leq 2^{3\lambda'_i}
\ee
since $[1,j_-]$ is contained in the $i$-th row and $|\alpha_i| = \lambda'_i-\mu'_i\leq\lambda'_i$.
This is the same bound as in~\cite{2212.06741}, but for all $b^2<0$.

By taking the product over~$i$ and accounting for the bound~\eqref{P0-bound} on~$P_0$, we deduce the bound announced in the lemma.  Of course, the constant $C>0$ may differ from that in~\eqref{P0-bound} due to the sub-exponential factors caused by shifting~$A,B$ at the outset.
\end{proof}

\begin{proposition}
\label{prop:SQCD-off}
For $b^2\in(-\infty,0)$, the vector multiplet contribution is bounded by its diagonal part, up to exponential factors: there exists $C>0$ such that (denoting by $n=|Y|$ the instanton number)
\be
|\zvec(a,Y)|
\leq 8^{2(N-1)n}e^{C\sqrt{n}(1+\log n)} \prod_{\alpha=1}^{N} \prod_{(i,j)\in Y_\alpha} \frac{1}{|E(0, Y_\alpha, Y_\alpha; i,j)|^{2N}} .
\ee
\end{proposition}

\begin{proof}
This is an immediate application of Lemma~\ref{lem:SQCD-off-diag} to bound each off-diagonal contribution $\prod_{(i,j)\in Y_\alpha}1/E(\dots)$ by the same product with $a_\alpha-a_\beta$ or $a_\alpha-a_\beta-\eps_1-\eps_2$ shifted to zero and $Y_\beta$ replaced by~$Y_\alpha$, and of Proposition~\ref{prop:irrelevant-masses} to half of the diagonal factors to shift $-\eps_1-\eps_2$ to~$0$ in $E(-\eps_1-\eps_2,Y_\alpha,Y_\alpha;i,j)$.
\end{proof}

\subsection{Fundamental matter: a bound on contents/hooks}

To conclude our study of SQCD with $b^2=\eps_1/\eps_2<0$, there remains to show that the contribution of fundamental matter, which involves the ``contents'' $(j-1)\eps_1+(i-1)\eps_2$ of boxes, is controlled by the diagonal vector multiplet contribution, which involves the ``hook'' $(\lambda'_i-j)\eps_1+(\lambda_j-i)\eps_2$.

\begin{lemma}\label{SQCD-bound-fun}
  For $b^2\in(-\infty,0)$ and $A\in\CC$ there exists a constant $C>0$ such that for any partition~$\lambda$ one has
  \bel{SQCD-emp-lambda}
  \prod_{(i,j)\in\lambda} \Bigl| \frac{E(A, \lambda, \emptyset; i,j)}{E(0, \lambda, \lambda; i,j)} \Bigr|
  \leq e^{C\sqrt{|\lambda|}(1+\log|\lambda|)} .
  \ee
\end{lemma}

\begin{proof}
  By Proposition~\ref{prop:irrelevant-masses}, it is enough to prove this bound for a specific $A\not\in\overline{\ZZ_{<0}\eps_1+\ZZ_{<0}\eps_2}$ which we select to be $A=\sqrt{-1}\eps_2$.
  It is convenient to divide the numerators and denominators by~$\eps_2$.  Then we reorganize the numerator by swapping $(i,j)\to (1+\lambda_j-i,j)$ and we bound it using the triangle inequality,
  \beal
  \prod_{(i,j)\in\lambda} \frac{1}{\eps_2} |E(A, \lambda, \emptyset; i,j)|
  & = \prod_{(i,j)\in\lambda} \bigl|A/\eps_2 + b^2 j + i\bigr|
  \\
  & \leq \prod_{(i,j)\in\lambda} \bigl( 1 + |b^2 j + i| \bigr) .
  \ee
  We shall show that, apart from $O(\sqrt{|\lambda|})$ exceptions, each factor here is bounded by a corresponding factor in the denominator, namely that
  \bel{sqcd-fun-ibbj}
  1 + |b^2 j + i|
  \leq
  - b^2 (\lambda'_{\iti} - \jt) + \lambda_{\jt} + 1 - \iti
  \ee
  for suitable $(\iti,\jt)$, which depends almost bijectively on $(i,j)$.  The relatively few failures of bijectivity are treated afterwards.

  To achieve~\eqref{sqcd-fun-ibbj}, we define two regions $\lambda^{\pm}$ that roughly partition~$\lambda$ according to the sign $b^2 j + i\gtrless 0$ (up to some exceptional boxes).
  The region~$\lambda^+$ is defined by $i\geq i_+ \coloneqq \floor{|b^2|j}$, and we take $(\iti,\jt) = (\lambda_j+i_+-i,j)$, which defines a bijection of~$\lambda^+$.  One checks that
  \be
  - b^2 (\lambda'_{\iti} - \jt) + \lambda_{\jt} + 1 - \iti
  \geq \lambda_j + 1 - \iti
  = i + 1 - i_+
  \geq i + 1 - |b^2| j
  \ee
  as desired.
  The region~$\lambda^-$ is defined by $j\geq j_- \coloneqq \floor{|b^{-2}|i}$ and we take $(\iti,\jt) = (i,\lambda'_i+j_--j)$, which defines a bijection of~$\lambda^-$.  One checks that, as desired,
  \be
  - b^2 (\lambda'_{\iti} - \jt) + \lambda_{\jt} + 1 - \iti
  \geq |b^2| (j-j_-) + 1
  \geq |b^2| j - i + 1 .
  \ee

  The two regions cover~$\lambda$ since for any box $i-|b^2|j$ is either non-negative, which implies $(i,j)\in\lambda^+$, or non-positive, which implies $(i,j)\in\lambda^-$, or both.  Due to the gap between $|b^2|j$ and its floor~$i_+$ and the gap between $|b^{-2}|i$ and its floor~$j_+$ the intersection $\lambda^+\cap\lambda^-$ is generally not empty.  It consists of some of the boxes for which $i-|b^2|j\in(-1,|b^2|)$.
  As explained below~\eqref{eq9olambda} there are $O(\sqrt{|\lambda|})$ boxes with $j\eps_1+i\eps_2$ in any given bounded region.  Thus, at most $O(\sqrt{|\lambda|})$ factors were double counted in our pairing argument above, which leads to the sub-exponential bound stated in the Lemma.
\end{proof}

We are reaching the end of the story for SQCD\@.
Recall from~\eqref{zfun} the contribution of $N_f$~fundamental hypermultiplets to the $U(N)$ instanton partition function:
\be
\zfun(Y)
= \prod_{\alpha=1}^N \prod_{f=1}^{N_f} \prod_{(i,j)\in Y_\alpha} E(a_\alpha - m_f , Y_\alpha, \emptyset; i,j) .
\ee
By Lemma~\ref{SQCD-bound-fun}, for some $C>0$ (rescaled by~$N_f$ compared to the lemma),
\be
|\zfun(Y)|
\leq \prod_{\alpha=1}^N \Bigl( e^{C\sqrt{|Y_\alpha|}(1+\log|Y_\alpha|)} \prod_{(i,j)\in Y_\alpha} |E(0, Y_\alpha, Y_\alpha; i,j)|^{N_f} \Bigr) ,
\ee
Combining this with Proposition~\ref{prop:SQCD-off} gives (with a different $C>0$)
\be
\bigl|\zfun(Y)\zvec(Y)\bigr|
\leq 8^{2(N-1)n} e^{C\sqrt{n}(1+\log n)} \prod_{\alpha=1}^N \prod_{(i,j)\in Y_\alpha} |E(0, Y_\alpha, Y_\alpha; i,j)|^{N_f-2N}
\ee
where $n=|Y|$.

For $N_f=2N$ this simplifies drastically and gives an exponential upper bound, which proves absolute convergence with radius
\be
R_N = 8^{-2(N-1)} > 0 .
\ee
Indeed,
\beal
|\Zinst|
& \leq \sum_Y |q|^{|Y|} |Z_Y(m,a)|
= \sum_Y |q|^{|Y|} \bigl|\zvec(a,Y) \zfun(a,Y;m)\bigr|
\\
& \leq \sum_n (|q|/R_N)^n e^{C\sqrt{n}(1+\log n)} p(n) .
\ee
Here, $p(n)$ is the number of integer partitions of~$n$, which is upper bounded by $\exp(O(\sqrt{n}))$, so the upper bound converges precisely for $|q|<R_N$.

For $N_f<2N$, the negative power of $|E(\dots)|$ makes the analogous bound converge for all~$q$.  Indeed, by Lemma~\ref{lem:bound-density}, all boxes have
\be
|E(0, Y_\alpha, Y_\alpha; i,j)| \geq \rho \coloneqq (2|q|/R_N)^{1/(2N-N_f)}
\ee
except for at most $C'\sqrt{|\lambda|}$ boxes for some $C'>0$.  The latter boxes are still bounded below since
\be
E(0, Y_\alpha, Y_\alpha; i,j)/\eps_2 = |b^2|((Y_\alpha')_i-j) + Y_{\alpha j}-i + 1 \geq 1.
\ee
Thus,
\beal
|\Zinst|
& \leq \sum_Y (|q|/R_N)^{|Y|} e^{C\sqrt{|Y|}(1+\log|Y|)} \prod_{\alpha=1}^N \prod_{(i,j)\in Y_\alpha} |E(0, Y_\alpha, Y_\alpha; i,j)|^{N_f-2N}
\\
& \leq \sum_Y (|q|/R_N)^{|Y|} e^{C\sqrt{|Y|}(1+\log|Y|)} \rho^{(N_f-2N)|Y|} |\eps_2/\rho|^{C'\sqrt{|Y|}}
\\
& \leq \sum_n 2^{-n} e^{C\sqrt{n}(1+\log n)} |\eps_2/\rho|^{C'\sqrt{n}} p(n) < +\infty .
\ee
This concludes the proof of Theorem~\ref{thm:SQCD} in the case $b^2\in(-\infty,0)$.

\section{Quivers with non-real \(b^2\)}
\label{sec:quivers}

Throughout this section we assume $b^2=\eps_1/\eps_2\in\CC\setminus\RR$.  The lattice $\Lambda(\eps_1,\eps_2) = \ZZ\eps_1+\ZZ\eps_2$ is then discrete.
A useful equivalence of norms is that there exists $C_{b^2}>1$ such that for any $k,l\in\ZZ$,
\bel{equiv-norm}
\frac{|\eps_2|}{C_{b^2}} \bigl(|k| + |l|\bigr) \leq |k\eps_1 + l\eps_2| \leq C_{b^2} |\eps_2| \bigl(|k|+|l|\bigr) .
\ee
This allows us to replace (up to constant loss of radius) all $E(\dots)$ factors by simpler integer factors.

\subsection{Fundamental matter: a bound on cohooks/hooks}

Our goal here is to bound the fundamental hypermultiplet contribution~$\zfun$ by (inverses of) diagonal vector multiplet contributions.
The main step is a comparison on hook and cohook products in Lemma~\ref{lem:cohook-hook}, which is essentially in the literature.  Then we deduce Proposition~\ref{prop:quiver-fun}

\begin{lemma}\label{lem:cohook-hook}
  For any partition~$\lambda$ the ratio of the cohook product by the hook product is bounded as
  \beal
  1 \leq \prod_{(i,j)\in\lambda} \frac{i+j-1}{h_\lambda(i,j)} \leq 2^{|\lambda|} ,
  \ee
  where $h_\lambda(i,j) = \lambda'_i-j + \lambda_j-i + 1$.
\end{lemma}

\begin{proof}
  This is an easy corollary of powerful asymptotic results on standard Young tableaux by Morales, Pak, and Panova.
  Based on Naruse's hook-length formula for the number of standard Young tableaux $e(\nu/\mu)$ of a given skew shape $\nu/\mu$ (difference of two Young diagrams $\mu\subseteq\nu$), that they prove, these authors bound that number as \cite[Theorem 3.3, Lemma 4.3]{1610.07561}
  \bel{eq43}
  1 \leq \frac{e(\nu/\mu)}{F(\nu/\mu)} \leq 2^{|\nu/\mu|} , \qquad
  F(\nu/\mu) \coloneqq \frac{|\nu/\mu|!}{\prod_{(i,j)\in\nu/\mu} h_\nu(i,j)} .
  \ee
  The product $F(\nu/\mu)$ reduces for $\mu=\emptyset$ to the usual hook-length formula for the dimension of the symmetric group representation labeled by~$\nu$.

  Following \cite[Section 12.1]{1610.07561}, we now take $\nu$ to be a $p\times q$ rectangle containing the partition $\lambda=(\lambda_1,\dots,\lambda_\ell)$, namely $p\geq\lambda_1$ and $q\geq\ell=\lambda'_1$, and take $\mu$ to be a $180$-degree rotation of $\nu\setminus\lambda$, namely $\mu=(p-\lambda_q,p-\lambda_{q-1},\dots,p-\lambda_1)$, so that $\nu/\mu$ is itself a $180$-degree rotation of~$\lambda$.
  The standard Young tableaux with shape $\lambda$ (labelings of boxes from $1$ to $|\lambda|$ that are increasing in every row and column) are in bijection with those with skew shape $\nu/\mu$ (simply apply $x\mapsto |\lambda|+1-x$ to all labels).
  Thus, $e(\nu/\mu)=e(\lambda)$.  This is in turn equal to $F(\lambda)$ by the usual hook-length formula.  We deduce
  \be
  1 \leq \frac{e(\nu/\mu)}{F(\nu/\mu)} = \frac{F(\lambda)}{F(\nu/\mu)} = \frac{\prod_{(i,j)\in\nu/\mu} h_\nu(i,j)}{\prod_{(i,j)\in\lambda} h_\lambda(i,j)}
  \leq 2^{|\lambda|} ,
  \ee
  which is the desired bound since $h_\nu(i,j)=(p+1-i)+(q+1-j)-1$ and boxes in $\nu/\mu$ and in $\lambda$ are in bijection under $(i,j)\mapsto(p+1-i,q+1-i)$.
\end{proof}

\begin{remark}
  (a) The lower bound, noticed already in \cite{1610.07561}, states that the product of hook lengths is less than that of cohook lengths $(i+j-1)$.
  In fact, the multiset of hook lengths majorizes that of cohook lengths~\cite{1903.11828} (same total sum, and greater or equal sums of the largest $p$ elements).

  (b) Numerics suggest that the constant $2$ in the upper bound is far from optimal, but the limit shape that extremizes the ratio solves a complicated integral equation that does not readily have an analytic solution.
\end{remark}

\begin{proposition}\label{prop:quiver-fun}
  Consider parameters $\eps_1,\eps_2\in\CC\setminus\{0\}$ with $b^2=\eps_1/\eps_2\in\CC\setminus\RR$, and consider $A\in\CC$.
  Then there exists constants $C_1,C_2>0$ with $C_1$ depending only on~$b^2$ and $C_2$~depending on $\eps_1,\eps_2,A$, such that for any partition~$\lambda$,
  \bel{complex-b2-fun-bound}
  \prod_{(i,j)\in\lambda} \Bigl| \frac{E(A, \lambda, \emptyset; i,j)}{E(0, \lambda, \lambda; i,j)} \Bigr|
  \leq e^{C_1|\lambda|} e^{C_2\sqrt{|\lambda|}(1+\log|\lambda|)} .
  \ee
  In particular, the contribution of the $M_I$ fundamental hypermultiplets of $U(N_I)$ in a term $(Y_1,\dots,Y_K)$ of the instanton partition function in \autoref{subsec:Zinst} is bounded as follows, with $n_I=|Y_I|$ the $U(N_I)$ instanton number:
  \bel{complex-b2-zfun}
  |\zfun(Y_I)|
  \leq e^{C_1M_In_I} e^{C_2M_IN_I\sqrt{n_I}(1+\log n_I)} \prod_{\alpha=1}^{N_I} \prod_{(i,j)\in Y_{I\alpha}} |E(0, Y_{I\alpha}, Y_{I\alpha}; i,j)|^{M_I} .
  \ee
\end{proposition}

\begin{proof}
  Thanks to Proposition~\ref{prop:irrelevant-masses}, we only need to prove~\eqref{complex-b2-fun-bound} for $A=-\eps_2$, which does not lie in the closed set $\ZZ_{<0}\eps_1+\ZZ_{<0}\eps_2$ for complex~$b^2$.
  By mapping $(i,j)\to(\lambda_j+1-i,j)$, the numerator product is equal to
  \be
  \prod_{(i,j)\in\lambda} E(-\eps_2, \lambda, \emptyset; i,j)
  = \prod_{(i,j)\in\lambda} \bigl(\eps_1 j + \eps_2 (i-1)\bigr) .
  \ee
  On the other hand, $E(0, \lambda, \lambda; i,j) = -\eps_1 (\lambda'_i-j) + \eps_2 (\lambda_j + 1 - i)$.
  By \eqref{equiv-norm}, we deduce (the factors of $|\eps_2|$ cancel out)
  \be
  \prod_{(i,j)\in\lambda} \Bigl| \frac{E(-\eps_2, \lambda, \emptyset; i,j)}{E(0, \lambda, \lambda; i,j)} \Bigr|
  \leq (C_{b^2})^{2|\lambda|}
  \prod_{(i,j)\in\lambda} \frac{j + i - 1}{h_\lambda(i,j)}
  \leq (2 C_{b^2}^2)^{|\lambda|} .
  \ee
  The bound~\eqref{complex-b2-zfun} on the $I$-th fundamental contribution is then immediate by bounding $|Y_{I\alpha}|\leq n_I$ for each~$\alpha$.
\end{proof}

\subsection{Vectors and bifundamentals: pairs of partitions}

Quivers also feature bifundamental hypermultiplets, whose contribution to $\Zinst$ involves pairs of partitions $Y_{I\alpha}$ and $Y_{J\beta}$ for edges $(I,J)\in V_1$ of the quiver.  Lemma~\ref{lem:bif-1} states an upper bound for a single such pair of partitions and Lemma~\ref{lem:bif-2} obtains a lower bound on a product of two such contributions.  Proposition~\ref{prop:bif} then assembles these into bounds on~$\zbif$ and~$\zvec$.  This enables us to finish proving Theorem~\ref{thm:quivers}, namely that instanton partition sums for unitary quivers converge absolutely in a product of disks (or of the whole complex plane in the asymptotically free case).

\begin{lemma}\label{lem:bif-1}
  For $b^2=\eps_1/\eps_2\in\CC\setminus\RR$, and $A\in\CC\setminus(\ZZ\eps_1+\ZZ_{<0}\eps_2)$, there exists constants $C_1,C_2>0$ with $C_1$ depending only on~$b^2$ and $C_2$~depending on $\eps_1,\eps_2,A$, such that for any partition~$\lambda,\mu$ (and denoting $p = |\lambda|+|\mu|$),
  \bel{complex-b2-bif-1}
  \prod_{(i,j)\in\lambda} \Bigl| \frac{E(A, \lambda, \mu; i,j)}{E(0, \lambda, \lambda; i,j)} \Bigr|
  \leq e^{C_1 p+C_2\sqrt{p}(1+\log p)} .
  \ee
\end{lemma}

\begin{proof}
As in previous proofs, Proposition~\ref{prop:irrelevant-masses} lets us focus on a single value of~$A$.  Since we have a pair of partitions here, this reference value must be taken to be $A \not\in \ZZ\eps_1+\ZZ_{<0}\eps_2$.  We choose $A=0$.

Thanks to~\eqref{equiv-norm}, the product of interest is bounded by the one without $b^2$ dependence:
\bel{eq410}
\biggl| \prod_{(i,j)\in\lambda} \frac{E(0, \lambda, \mu; i,j)}{E(0, \lambda, \lambda; i,j)} \biggr|
\leq (C_{b^2})^{2|\lambda|} \prod_{(i,j)\in\lambda} \frac{|\mu'_i - j| + \lambda_j + 1 - i}{\lambda'_i - j + \lambda_j + 1 - i} .
\ee
Here we included absolute values for $\mu'_i - j$ since that can have either sign.
The denominator is $h_\lambda(i,j)$.

Next, we introduce the Young diagram $\lambda\cap\mu$ consisting of boxes that are in both $\lambda$ and~$\mu$:
\be
(\lambda\cap\mu)_j = \min(\lambda_j,\mu_j) , \qquad
(\lambda\cap\mu)'_i = \min(\lambda'_i,\mu'_i) .
\ee
We shall treat separately the products over boxes in $\lambda\cap\mu$ and those in $\lambda\setminus(\lambda\cap\mu)$.

Consider the set of boxes $(i,j)\in\lambda\setminus(\lambda\cap\mu)$, which is equivalent to $\mu'_i<j$ or equivalently $\mu_j<i$.  The numerator corresponding to such a box is $(j - \mu'_i + \lambda_j + 1 - i)$, which is the semi-hook in the skew diagram $\lambda\setminus(\lambda\cap\mu)$. The product of semi-hooks is known to bound the hook product and be bounded by an exponential multiple thereof, so
\be
\prod_{(i,j)\in \lambda\setminus(\lambda\cap\mu)} (|\mu'_i - j| + \lambda_j + 1 - i)
\leq e^{C|\lambda|} \prod_{(i,j)\in \lambda\setminus(\lambda\cap\mu)} (\lambda'_i - j + \lambda_j + 1 - i) .
\ee

Consider next the boxes in $\lambda\cap\mu$.  Note that $\mu'_i-j\geq 0$ for these boxes, which allows us to remove all absolute values in~\eqref{eq410} and ask for a bound on
\bel{eq417}
\prod_{(i,j)\in\lambda\cap\mu} \frac{\mu'_i - j + \lambda_j + 1 - i}{\lambda'_i - j + \lambda_j + 1 - i}
= \prod_{(i,j)\in\lambda\cap\mu} \Bigl( 1 + \frac{\mu'_i - \lambda'_i}{h_\lambda(i,j)} \Bigr) .
\ee
For each row~$i$, if $\mu'_i\leq\lambda'_i$ the factor in parentheses is in $\leq 1$.  Thus, the contribution of all such rows is $\leq 1$.
If instead $\mu'_i>\lambda'_i$ we reuse the same argument as in~\eqref{Pmi-bound} to get a binomial coefficient upper bound since $h_\lambda(i,j)$ is strictly monotonically decreasing in~$j$ with integer values:
\be
\prod_{j=1}^{\lambda'_i} \frac{\mu'_i - j + \lambda_j + 1 - i}{\lambda'_i - j + \lambda_j + 1 - i}
\leq \prod_{d=1}^{\lambda'_i} \frac{d + \mu'_i - \lambda'_i}{d} = \binom{\mu'_i}{\lambda'_i} \leq 2^{\mu'_i} ,
\qquad \text{when $\lambda'_i < \mu'_i$}.
\ee
Altogether this gives the desired upper bound.
\end{proof}

\begin{lemma}\label{lem:bif-2}
  For $b^2=\eps_1/\eps_2\in\CC\setminus\RR$, and $A\in\CC\setminus(\ZZ\eps_1+\ZZ_{<0}\eps_2)$, there exists constants $C_1,C_2>0$ with $C_1$ depending only on~$b^2$ and $C_2$~depending on $\eps_1,\eps_2,A$, such that for any partition~$\lambda,\mu$ (and denoting $p = |\lambda|+|\mu|$),
  \bel{complex-b2-bif-2}
  \prod_{(i,j)\in\lambda} \Bigl| \frac{E(0, \lambda, \lambda; i,j)}{E(A, \lambda, \mu; i,j)} \Bigr|
  \leq e^{C_1 p+C_2\sqrt{p}(1+\log p)} .
  \ee
\end{lemma}

\begin{proof}
Note that this is the inverse ratio as in Lemma~\ref{lem:bif-1}.
The proof is similar to Lemma~\ref{lem:bif-1}: we eliminate the $b^2$ dependence and we split boxes into $\lambda\cap\mu$ and $\lambda\setminus(\lambda\cap\mu)$.
For boxes in $\lambda\cap\mu$ the bound by binomial coefficients is easy to get with the same strategy as for Lemma~\ref{lem:bif-1}.  This gives
\be
\prod_{j=1}^{\lambda'_i} \frac{\lambda'_i - j + \lambda_j + 1 - i}{\mu'_i - j + \lambda_j + 1 - i}
\leq \binom{\max(\mu'_i,\lambda'_i)}{\mu'_i} \leq 2^{\lambda'_i} .
\ee

Consider the set of boxes $(i,j)\in\lambda\setminus(\lambda\cap\mu)$, which is equivalent to $\mu'_i<j$ or equivalently $\mu_j<i$.  In a given column~$j$, we consider two boxes $(i,j)$ and $(\iti,j)$ with
\be
\iti = \lambda_j + \mu_j + 1 - i .
\ee
Both boxes are in $\lambda\setminus(\lambda\cap\mu)$.
The product of denominators corresponding to these two boxes is
\beal
& (|\mu'_i - j| + \lambda_j + 1 - i) (|\mu'_{\iti} - j| + \lambda_j + 1 - \iti)
\\
& \quad = (j-\mu'_i + \iti - \mu_j) (j-\mu'_{\iti} + i - \mu_j)
\\
& \quad \geq (j-\mu'_i + i - \mu_j) (j-\mu'_{\iti} + \iti - \mu_j)
\ee
where we used that $(a+b)(\at+\bt) \leq (a+\bt)(\at+b)$ for $0\leq a\leq \at$ and $0\leq b\leq\bt$.
By taking a product over boxes in $\lambda\setminus(\lambda\cap\mu)$ (and taking the square root to avoid double-counting), we obtain
\bel{eq414}
\prod_{(i,j)\in\lambda\setminus(\lambda\cap\mu)} \bigl(|\mu'_i - j| + \lambda_j + 1 - i\bigr)
\geq \prod_{(i,j)\in\lambda\setminus(\lambda\cap\mu)} (j-\mu'_i + i - \mu_j) .
\ee
For these boxes $(i,j)$, one has $\mu'_i = (\lambda\cap\mu)'_i$ and $\mu_j = (\lambda\cap\mu)_j$, so the last product is a product of (one plus) cohooks in the skew Young diagram $\lambda\setminus(\lambda\cap\mu)$.  To be precise, cohooks are $j-\mu'_i + i - \mu_j - 1$ instead, which is smaller.

We now deploy the same argument as below~\eqref{eq43}.  The number $e(\lambda\setminus(\lambda\cap\mu))$ of standard Young tableaux with skew shape $\lambda\setminus(\lambda\cap\mu)$ and with its $180$-degree rotation are the same.  They differ from the corresponding $F(\dots)$ by a factor of $2^{\pm t}$ so the corresponding $F(\dots)$ differ by at most a factor of~$4^t$.  This means that the product of hooks in $\lambda\setminus(\lambda\cap\mu)$ differs from the product of cohooks by~$4^t$.  Since the hook of a given box $(i,j)$ in the skew shape matches the hook in~$\lambda$, we get
\be
\prod_{(i,j)\in\lambda\setminus(\lambda\cap\mu)} \bigl(|\mu'_i - j| + \lambda_j + 1 - i\bigr)
\geq 2^{-2t} \prod_{(i,j)\in\lambda\setminus(\lambda\cap\mu)} h_\lambda(i,j) .
\ee
\end{proof}

\begin{proposition}\label{prop:bif}
  For $b^2=\eps_1/\eps_2\in\CC\setminus\RR$, the contribution of the vector multiplet of $U(N_I)$ in a term $(Y_1,\dots,Y_K)$ of the instanton partition function in \autoref{subsec:Zinst} is bounded as
  \beall{complex-b2-zvec}
  & |\zvec(a_I,Y_I)|
 \leq e^{C_1 p} e^{C_2\sqrt{p}(1+\log p)} \prod_{\alpha=1}^{N_I} \prod_{(i,j)\in Y_{I\alpha}} |E(0, Y_{I\alpha}, Y_{I\alpha}; i,j)|^{-2N_I} ,
  \ee
  and the contribution of the bifundamental hypermultiplet of $U(N_I)$ and $U(N_J)$ (or adjoint hypermultiplet for $I=J$) is bounded as
  \beall{complex-b2-zbif}
  & |\zbif(a_I,Y_I,a_J,Y_J;m)|
  \\
  & \quad \leq e^{C_1 p} e^{C_2\sqrt{p}(1+\log p)} \prod_{1\leq\alpha\leq N_I} \prod_{(i,j)\in Y_{I\alpha}} |E(0, Y_{I\alpha}, Y_{I\alpha}; i,j)|^{N_J}
  \\
  & \qquad\qquad\qquad\qquad\quad \times \prod_{1\leq\beta\leq N_J} \prod_{(i,j)\in Y_{J\beta}} |E(0, Y_{J\beta}, Y_{J\beta}; i,j)|^{N_I} ,
  \ee
  where $p=|Y_I|+|Y_J|$, the constant $C_1>0$ depends on $b^2,N_I,N_J$, and the constant $C_2>0$ depends on $N_I,N_J,\eps_1,\eps_2$ and on masses and Coulomb branch parameters, but not on the Young diagrams.
\end{proposition}

\begin{proof}
  This is an immediate application of Lemmas~\ref{lem:bif-1} and~\ref{lem:bif-2}.
\end{proof}

\paragraph{Finishing the proof of Theorem~\ref{thm:quivers}.}

In Propositions~\ref{prop:quiver-fun} and~\ref{prop:bif} we have proven upper bounds on the fundamental, bifundamental, and vector multiplet contributions.  Now we simply assemble them: denoting $n=|Y_1|+\dots+|Y_K|$ the total instanton number, we find that there are constants $C_1,C_2>0$ with $C_1$ depending only on~$b^2$ and the quiver, and $C_2$ depending in addition on mass/Coulomb parameters, such that
\beal
|\Zinst|
& \leq \sum_{Y_1,\dots,Y_K} |q_1|^{|Y_1|} \dotsm |q_K|^{|Y_K|} |Z_Y(m,a)|
\\
& \leq \sum_{Y_1,\dots,Y_K} |q_1|^{|Y_1|} \dotsm |q_K|^{|Y_K|} e^{C_1 n} e^{C_2 \sqrt{n} (1+\log n)} \\[-2ex]
& \qquad\qquad\qquad\qquad \times \prod_{I=1}^K \prod_{\alpha=1}^{N_I} \prod_{(i,j)\in Y_{I\alpha}} |E(0,Y_{I\alpha},Y_{I\alpha};i,j)|^{-b_I} ,
\ee
where $b_I$ is the balance of node~$I$ defined in~\eqref{balance-ineq}.
The theories of interest have $b_I\geq 0$.
Assume now that for each node~$I$,
\be
b_I > 0 \text{ or } |q_I| < e^{-C_1} .
\ee
Then this series converges by the same argument as the end of~\autoref{sec:SQCD-neg}.
Indeed, the factors $|E(0,Y_{I\alpha},Y_{I\alpha};i,j)|^{-b_I}$ of unbalanced nodes decay superexponentially in the instanton number $n_I=|Y_I|$, which compensates the power $|q_I|^{n_I} e^{C_1 n_I}$ regardless of~$q_I$, while for balanced nodes the factor $|q_I|^{n_I} e^{C_1 n_I} = |q_I e^{C_1}|^n$ gives a geometric series.

\section{Outlook}

The results established in this paper with coarse combinatorial techniques leave a few questions open about the class of quiver theories considered here.
Concerning SQCD, it seems possible to show absolute convergence with the optimal radius $R=1$ for $b^2<0$ by improving the comparison of off-diagonal and diagonal vector multiplet contributions in \autoref{sec:SQCD-off-diag}, by combining the $(\alpha,\beta)$ and $(\beta,\alpha)$ off-diagonal contributions and finding a better bijection of boxes.
For non-real~$b^2$, we should expect only conditional convergence with radius~$1$, but absolute convergence only in a smaller disk: a brief numerical exploration suggests that large partitions with certain limit shapes that extremize $Z_Y(m,a)^{1/|Y|}$ (which should be a restriction of those in Nekrasov--Okounkov~\cite{hep-th/0306238}) force the absolute convergence radius to depend on~$b^2$.

Quivers with unitary gauge groups should behave as SQCD\@.  To get a positive radius for real $b^2<0$, we need an upper bound on the bifundamental hypermultiplet contribution, which amounts to finding a lower bound for the ratio in Lemma~\ref{lem:SQCD-off-diag}.  In contrast to the upper bound, the lower bound inevitably depends on both partitions (as seen by considering the large~$\mu$ limit), which may lead to difficulties.

Similar techniques should prove convergence for all Lagrangian theories classified in~\cite{1309.5160} for which Nekrasov partition functions are known as explicit sums over Young diagrams.  The density bounds of Lemma~\ref{lem:bound-density} for instance may be usefully rephrased as a lack of concentration of eigenvalues in the matrix models describing the $n$-instanton partition function.  A potential difficulty for gauge groups other than $U(N)$ is the presence of higher-order poles in the matrix model description.

A major remaining challenge is to understand convergence of the instanton series~$\Zinst$ in the case $b^2>0$, which appears to be entirely open, except for the 4d $\Nsusy=2^*$ theory~\cite{2602.19425,2003.03802}.  This is the most natural regime from the point of view of the AGT correspondence, since it corresponds to conformal blocks in unitary Liouville/Toda CFTs.
The crucial question is whether summing over all tuples with $k$ boxes cures the runaway growth of some contributions.
One hope would be to try and rephrase the bounds at the level of the rank~$n$ matrix integral expression of~$Z_n$, rather than at the level of the sum over poles (labelled by tuples of Young diagrams).

Another approach could be to start with a non-zero imaginary part $\Im(b^2)$ where arbitrary reorderings are possible.  It may be possible to resum the series over Young diagrams (e.g., through a relation with instanton-vortex partition functions) and find an upper bound on $|\Zinst|$ that is uniformly bounded as $b^2$ tends to a specific value in $(0,+\infty)$, for fixed $|q|<1$.

The K-theoretic Nekrasov partition function of five-dimensional uplifts of these theories are also of interest.

\section*{Acknowledgements}

The author thanks Sylvain Lacroix for discussions, and Fabrizio Del Monte, Harini Desiraju, Alba Grassi, and Vincent Vargas for organizing the ``Conformal field theory 3 ways: integrable, probabilistic, and supersymmetric'' SwissMAP conference in January 2024, which reignited his interest into these topics.

\bibliographystyle{plain}
\bibliography{ZinstRadius}

\end{document}